\documentclass{article}%
\usepackage{amssymb}
\usepackage{amsmath}
\usepackage{amsfonts}
\usepackage{graphicx}%
\providecommand{\U}[1]{\protect\rule{.1in}{.1in}}
\newtheorem{theorem}{Theorem}[section]

\newtheorem{conjecture}[theorem]{Conjecture}
\newtheorem{corollary}[theorem]{Corollary}

\newtheorem{definition}[theorem]{Definition}
\newtheorem{example}[theorem]{Example}

\newtheorem{lemma}[theorem]{Lemma}

\newtheorem{problem}[theorem]{Problem}
\newtheorem{proposition}[theorem]{Proposition}
\newtheorem{remark}[theorem]{Remark}

\newenvironment{proof}[1][Proof]{\noindent\textbf{#1.} }{\ \rule{0.5em}{0.5em}}
\begin{document}

\title{Monotonic Properties of Collections of Maximum Independent Sets of a Graph}
\author{Adi Jarden\\Department of Mathematics and Computer Science\\Ariel University, Israel\\jardena@ariel.ac.il
\and Vadim E. Levit\\Department of Mathematics and Computer Science\\Ariel University, Israel\\levitv@ariel.ac.il
\and Eugen Mandrescu\\Department of Computer Science\\Holon Institute of Technology, Israel\\eugen\_m@hit.ac.il}
\date{}
\maketitle

\begin{abstract}
Let $G$ be a simple graph with vertex set $V\left(  G\right)  $. A set
$S\subseteq V\left(  G\right)  $ is \textit{independent} if no two vertices
from $S$ are adjacent. The graph $G$ is known to be a K\"{o}nig-Egerv\'{a}ry
if $\alpha\left(  G\right)  +\mu\left(  G\right)  =$ $\left\vert V\left(
G\right)  \right\vert $, where $\alpha\left(  G\right)  $ denotes the size of
a maximum independent set and $\mu\left(  G\right)  $ is the cardinality of a
maximum matching.

Let $\Omega(G)$ denote the family of all maximum independent sets, and $f$ be
the function from subcollections $\Gamma$ of $\Omega(G)$ to $\mathbb{N}$ such
that $f(\Gamma)=|\bigcup\Gamma|+|\bigcap\Gamma|$. Our main finding claims that
$f$ is $\vartriangleleft$-increasing, where the preorder $\Gamma^{\prime
}\vartriangleleft\Gamma$ means that $\bigcup\Gamma^{\prime}\subseteq
\bigcup\Gamma$ and $\bigcap\Gamma\subseteq\bigcap\Gamma^{\prime}$. Let us say
that a family $\emptyset\neq\Gamma\subseteq\Omega\left(  G\right)  $ is a
K\"{o}nig-Egerv\'{a}ry collection if $\left\vert \bigcup\Gamma\right\vert
+\left\vert \bigcap\Gamma\right\vert =2\alpha(G)$. We conclude with the
observation that for every graph $G$ each subcollection of a
K\"{o}nig-Egerv\'{a}ry collection is K\"{o}nig-Egerv\'{a}ry as well.


\end{abstract}

\section{Introduction}

Throughout this paper $G$ is a finite simple graph with vertex set $V(G)$ and
edge set $E(G)$. If $X\subseteq V\left(  G\right)  $, then $G[X]$ is the
subgraph of $G$ induced by $X$. By $G-W$ we mean either the subgraph
$G[V\left(  G\right)  -W]$, if $W\subseteq V(G)$, or the subgraph obtained by
deleting the edge set $W$, for $W\subseteq E(G)$. In either case, we use
$G-w$, whenever $W$ $=\{w\}$. If $A,B$ $\subseteq V\left(  G\right)  $, then
$(A,B)$ stands for the set $\{ab:a\in A,b\in B,ab\in E\left(  G\right)  \}$.

The \textit{neighborhood} $N(v)$ of a vertex $v\in V\left(  G\right)  $ is the
set $\{w:w\in V\left(  G\right)  $ \textit{and} $vw\in E\left(  G\right)  \}$;
in order to avoid ambiguity, we use also $N_{G}(v)$ instead of $N(v)$. The
\textit{neighborhood} $N(A)$ of $A\subseteq V\left(  G\right)  $ is $\{v\in
V\left(  G\right)  :N(v)\cap A\neq\emptyset\}$, and $N[A]=N(A)\cup A$.

A set $S\subseteq V(G)$ is \textit{independent} if no two vertices from $S$
are adjacent, and by $\mathrm{Ind}(G)$ we mean the family of all the
independent sets of $G$. An independent set of maximum size is a
\textit{maximum independent set} of $G$, and $\alpha(G)=\max\{\left\vert
S\right\vert :S\in\mathrm{Ind}(G)\}$.

Let $\Omega(G)$ denote the family of all maximum independent sets,
$\mathrm{core}(G)=%
{\displaystyle\bigcap}
\{S:S\in\Omega(G)\}$ \cite{LevMan2002a}, and $\mathrm{corona}(G)=\cup
\{S:S\in\Omega(G)\}$ \cite{BorosGolLev}.

A \textit{matching} is a set $M$ of pairwise non-incident edges of $G$. If
$A\subseteq V(G)$, then $M\left(  A\right)  $ is the set of all the vertices
matched by $M$ with vertices belonging to $A$. A matching of maximum
cardinality, denoted $\mu(G)$, is a \textit{maximum matching}. For every
matching $M$, we denote the set of all vertices that $M$ saturates by
$V\left(  M\right)  $, and by $M(x)$ we denote the vertex $y$ satisfying
$xy\in M$.

For $X\subseteq V(G)$, the number $\left\vert X\right\vert -\left\vert
N(X)\right\vert $ is the \textit{difference} of $X$, denoted $d(X)$. The
\textit{critical difference} $d(G)$ is $\max\{d(X):X\subseteq V(G)\}$. The
number $\max\{d(I):I\in\mathrm{Ind}(G)\}$ is the \textit{critical independence
difference} of $G$, denoted $id(G)$. Clearly, $d(G)\geq id(G)$. It was shown
in \cite{Zhang1990} that $d(G)$ $=id(G)$ holds for every graph $G$. If $A$ is
an independent set in $G$ with $d\left(  X\right)  =id(G)$, then $A$ is a
\textit{critical independent set} \cite{Zhang1990}.

\begin{theorem}
\label{th3}\cite{ButTruk2007} Each critical independent set can be enlarged to
a maximum independent set.
\end{theorem}

\begin{theorem}
\label{th4}\cite{LevMan2012a} For a graph $G$, the following assertions are true:

\emph{(i)} $\mathrm{\ker}(G)\subseteq\mathrm{core}(G)$;

\emph{(ii)} if $A$ and $B$ are critical in $G$, then $A\cup B$ and $A\cap B$
are critical as well;

\emph{(iii)} $G$ has a unique minimal independent critical set, namely,
$\mathrm{\ker}(G)$.
\end{theorem}

It is well-known that $\alpha(G)+\mu(G)\leq\left\vert V(G)\right\vert $ holds
for every graph $G$. Recall that if $\alpha(G)+\mu(G)=\left\vert
V(G)\right\vert $, then $G$ is a \textit{K\"{o}nig-Egerv\'{a}ry graph}
\cite{Deming1979,Sterboul1979}. For example, each bipartite graph is a
K\"{o}nig-Egerv\'{a}ry graph as well. Various properties of
K\"{o}nig-Egerv\'{a}ry graphs can be found in
\cite{Korach2006,levm4,LevMan2013b}.

A proof of a conjecture of Graffiti.pc \cite{DeLaVina} yields a new
characterization of K\"{o}nig-Egerv\'{a}ry graphs: these are exactly the
graphs having a critical maximum independent set \cite{Larson2011}.

\begin{theorem}
\label{th5}\cite{LevMan2012b} For a graph $G$, the following assertions are equivalent:

\emph{(i)} $G$ is a K\"{o}nig-Egerv\'{a}ry graph;

\emph{(ii)} there exists some maximum independent set which is critical;

\emph{(iii)} each of its maximum independent sets is critical.
\end{theorem}

For a graph $G$, let us denote
\begin{align*}
\mathrm{\ker}(G)  &  =%
{\displaystyle\bigcap}
\left\{  A:A\text{ \textit{is a critical independent set}}\right\}  \text{
\cite{LevMan2012a},}\\
\mathrm{MaxCritIndep}(G)  &  =\left\{  S:S\text{ \textit{is a maximum critical
independent set}}\right\} \\
\mathrm{diadem}(G)  &  =%
{\displaystyle\bigcup}
\mathrm{MaxCritIndep}(G)\text{, and }\mathrm{nucleus}(G)=%
{\displaystyle\bigcap}
\mathrm{MaxCritIndep}(G)\text{.}%
\end{align*}

Clearly, $\mathrm{\ker}(G)\subseteq\mathrm{nucleus}(G)$ holds for every graph
$G$. In addition, by Theorem \ref{th3}, the inclusion $\mathrm{diadem}%
(G)\subseteq\mathrm{corona}(G)$ is true for every graph $G$.

In \cite{LevManLemma2011} the following lemma was introduced.

\begin{lemma}
[Matching Lemma]\cite{LevManLemma2011} \label{MatchingLemma}If $A\in
\mathrm{Ind}(G),\Lambda\subseteq\Omega(G)$, and $\left\vert \Lambda\right\vert
\geq1$, then there exists a matching from $A-%
{\displaystyle\bigcap}
\Lambda$ into $%
{\displaystyle\bigcup}
\Lambda-A$.
\end{lemma}

\section{Monotonicity results}

We define the following preorder, denoted $\vartriangleleft$, on the class of
collections of sets.

\begin{definition}
Let $\Gamma,\Gamma^{\prime}$ be two collections of sets. We write
$\Gamma^{\prime}\vartriangleleft\Gamma$ if $%
{\displaystyle\bigcup}
\Gamma^{\prime}\subseteq%
{\displaystyle\bigcup}
\Gamma$ and $%
{\displaystyle\bigcap}
\Gamma\subseteq%
{\displaystyle\bigcap}
\Gamma^{\prime}$.
\end{definition}

\begin{example}
\emph{(i) }$\{\{1\}\}\ntriangleleft\{\{1,2\}\}$

\emph{(ii) }$\{\{1,2\},\{2,3\}\}\ntriangleleft\{\{1,2\},\{1,3\}\}$,

\emph{(iii) }$\{\{1,2\},\{2,3\}\}\vartriangleleft\{\{1,2\},\{1,3\},\{2,3\}\}$,

\emph{(iv) }$\{\{1,2\},\{2,3\}\}\vartriangleleft\{\{1,2\},\{1,3\},\{2\}\}$.

\emph{(v) }$\mathrm{MaxCritIndep}(G)\vartriangleleft$ $\Omega(G)$ is true for
every bipartite graph $G$, since $\mathrm{core}(G)\subseteq\ker(G)$ and
\textrm{diadem}$(G)\subseteq\mathrm{corona}(G)$.
\end{example}

\begin{theorem}
\label{matching}If $\emptyset\neq\Gamma\subseteq\Omega(G)$, and $\emptyset
\neq\Gamma^{\prime}\subseteq\mathrm{Ind}(G)$, then there is a matching from $%
{\displaystyle\bigcap}
\Gamma^{\prime}-%
{\displaystyle\bigcap}
\Gamma$ into $%
{\displaystyle\bigcup}
\Gamma-%
{\displaystyle\bigcup}
\Gamma^{\prime}$.
\end{theorem}

\begin{proof}
Let $S=%
{\displaystyle\bigcap}
\Gamma^{\prime}$. Since $S$ is independent, by Lemma \ref{MatchingLemma},
there is a matching $M$ from $S-%
{\displaystyle\bigcap}
\Gamma$ into $%
{\displaystyle\bigcup}
\Gamma-S$. For each $x\in S-%
{\displaystyle\bigcap}
\Gamma$, we have $M(x)\notin%
{\displaystyle\bigcup}
\Gamma^{\prime}$, because every $A\in\Gamma^{\prime}$ is independent.
Consequently, $M$ is a matching from $%
{\displaystyle\bigcap}
\Gamma^{\prime}-%
{\displaystyle\bigcap}
\Gamma$ into $%
{\displaystyle\bigcup}
\Gamma-%
{\displaystyle\bigcup}
\Gamma^{\prime}$, as claimed.
\end{proof}

Choosing $\Gamma=\Omega(G)$ in Theorem \ref{matching}, we get the following.

\begin{corollary}
\label{corollary 5}If $\emptyset\neq\Gamma^{\prime}\subseteq\mathrm{Ind}(G)$,
then there is a matching from $%
{\displaystyle\bigcap}
\Gamma^{\prime}-\mathrm{core}(G)$ into $\mathrm{corona}(G)-%
{\displaystyle\bigcup}
\Gamma^{\prime}$.
\end{corollary}

Choosing $\Gamma^{\prime}=\left\{  S\right\}  \subseteq\Omega(G)$ in Corollary
\ref{corollary 5}, we get the following.

\begin{corollary}
\cite{BorosGolLev} For every graph $G$ and for every $S\in\Omega(G)$, there is
a matching from $S-\mathrm{core}(G)$ into $\mathrm{corona}(G)-S$.
\end{corollary}

\begin{theorem}
\label{the main theorem}If $\Gamma\subseteq\Omega(G)$ and $\Gamma^{\prime
}\subseteq\mathrm{Ind}(G)$ is such that $\Gamma^{\prime}\vartriangleleft
\Gamma$, then
\[
\left\vert \bigcup\Gamma^{\prime}\right\vert +\left\vert
{\displaystyle\bigcap}
\Gamma^{\prime}\right\vert \leq\left\vert \bigcup\Gamma\right\vert
+\left\vert
{\displaystyle\bigcap}
\Gamma\right\vert .
\]
In particular, $f:\left\{  \Gamma:\Gamma\subseteq\Omega(G)\right\}
\longrightarrow%
\mathbb{N}
,\ f\left(  \Gamma\right)  =\left\vert
{\displaystyle\bigcup}
\Gamma\right\vert +\left\vert
{\displaystyle\bigcap}
\Gamma\right\vert $ is $\vartriangleleft$-increasing.
\end{theorem}

\begin{proof}
If $\Gamma^{\prime}=\emptyset$ or $\Gamma=\emptyset$, then the inequality
clearly holds. Otherwise, according to Theorem \ref{matching}, there is a
matching $M$ from $\bigcap\Gamma^{\prime}-\bigcap\Gamma$\ into $\bigcup
\Gamma-\bigcup\Gamma^{\prime}$. Thus
\[
\left\vert \bigcap\Gamma^{\prime}-\bigcap\Gamma\right\vert \leq\left\vert
\bigcup\Gamma-\bigcup\Gamma^{\prime}\right\vert .
\]
Since $%
{\displaystyle\bigcap}
\Gamma\subseteq%
{\displaystyle\bigcap}
\Gamma^{\prime}$ and $%
{\displaystyle\bigcup}
\Gamma^{\prime}\subseteq%
{\displaystyle\bigcup}
\Gamma$, we have
\[
\left\vert \bigcap\Gamma^{\prime}-\bigcap\Gamma\right\vert =\left\vert
\bigcap\Gamma^{\prime}\right\vert -\left\vert \bigcap\Gamma\right\vert \text{,
and }\left\vert \bigcup\Gamma-\bigcup\Gamma^{\prime}\right\vert =\left\vert
\bigcup\Gamma\right\vert -\left\vert \bigcup\Gamma^{\prime}\right\vert ,
\]
which completes the proof.
\end{proof}

\begin{corollary}
\label{corollary 1}If $\Gamma^{\prime}\subseteq\Gamma\subseteq$ $\Omega(G)$,
then $\left\vert
{\displaystyle\bigcup}
\Gamma^{\prime}\right\vert +\left\vert
{\displaystyle\bigcap}
\Gamma^{\prime}\right\vert \leq\left\vert
{\displaystyle\bigcup}
\Gamma\right\vert +\left\vert
{\displaystyle\bigcap}
\Gamma\right\vert $.
\end{corollary}

\begin{proof}
It follows immediately by Theorem \ref{the main theorem}, because
$\Gamma^{\prime}\subseteq\Gamma$ implies $\Gamma^{\prime}\vartriangleleft
\Gamma$.
\end{proof}

\begin{corollary}
$\left\vert \mathrm{corona}(G)\right\vert +\left\vert \mathrm{core}%
(G)\right\vert =2\alpha(G)$ if and only if $\left\vert
{\displaystyle\bigcup}
\Gamma\right\vert +\left\vert
{\displaystyle\bigcap}
\Gamma\right\vert =2\alpha(G)$ holds for each non-empty $\Gamma\subseteq
\Omega\left(  G\right)  $.
\end{corollary}

\begin{corollary}
\cite{LevManLemma2011}\label{corollary 8} If $\Gamma\subseteq\Omega\left(
G\right)  ,\left\vert \Gamma\right\vert \geq1$, then $2\alpha(G)\leq
\left\vert
{\displaystyle\bigcup}
\Gamma\right\vert +\left\vert
{\displaystyle\bigcap}
\Gamma\right\vert $.
\end{corollary}

Let us consider the graphs $G_{1}$ and $G_{2}$ from Figure \ref{fig2222}:
$\mathrm{core}(G_{1})=\left\{  a,b,c,d\right\}  $ and it is a critical set,
while $\mathrm{core}(G_{2})=\left\{  x,y,z,w\right\}  $ and it is not
critical.\begin{figure}[h]
\setlength{\unitlength}{1cm}\begin{picture}(5,1.8)\thicklines
\multiput(1,0.5)(1,0){7}{\circle*{0.29}}
\multiput(1,1.5)(1,0){6}{\circle*{0.29}}
\put(1,0.5){\line(1,0){6}}
\put(1,1.5){\line(1,-1){1}}
\put(2,0.5){\line(0,1){1}}
\put(3,1.5){\line(1,-1){1}}
\put(3,1.5){\line(1,0){1}}
\put(4,0.5){\line(0,1){1}}
\put(5,0.5){\line(0,1){1}}
\put(6,1.5){\line(1,-1){1}}
\put(6,0.5){\line(0,1){1}}
\put(0.7,1.5){\makebox(0,0){$a$}}
\put(0.7,0.5){\makebox(0,0){$b$}}
\put(1.7,1.5){\makebox(0,0){$c$}}
\put(2.7,1.5){\makebox(0,0){$e$}}
\put(4.3,1.5){\makebox(0,0){$f$}}
\put(5.3,1.5){\makebox(0,0){$g$}}
\put(3,0.85){\makebox(0,0){$d$}}
\put(4,0){\makebox(0,0){$G_{1}$}}
\multiput(8,0.5)(1,0){6}{\circle*{0.29}}
\multiput(8,1.5)(1,0){5}{\circle*{0.29}}
\put(8,0.5){\line(1,0){5}}
\put(8,1.5){\line(1,-1){1}}
\put(9,0.5){\line(0,1){1}}
\put(10,0.5){\line(0,1){1}}
\put(10,0.5){\line(1,1){1}}
\put(10,1.5){\line(1,0){1}}
\put(12,0.5){\line(0,1){1}}
\put(12,1.5){\line(1,-1){1}}
\put(8.3,1.5){\makebox(0,0){$x$}}
\put(8.2,0.75){\makebox(0,0){$y$}}
\put(9.3,1.5){\makebox(0,0){$z$}}
\put(11,0.8){\makebox(0,0){$w$}}
\put(10,0){\makebox(0,0){$G_{2}$}}
\end{picture}\caption{Both $G_{1}$ and $G_{2}$\ are not K\"{o}nig-Egerv\'{a}ry
graphs.}%
\label{fig2222}%
\end{figure}
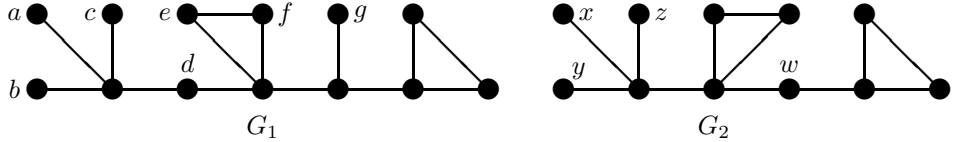

Moreover, $\mathrm{\ker}(G_{1})=\left\{  a,b,c\right\}  \subset\mathrm{core}%
(G_{1})\subset\left\{  a,b,c,d,g\right\}  =\mathrm{nucleus}(G_{1})$, where
$\mathrm{nucleus}(G_{1})=$ $A_{1}\cap A_{2}$, and $A_{1}=\left\{
a,b,c,d,e,g\right\}  $ and $A_{2}=\left\{  a,b,c,d,f,g\right\}  $ are all the
maximum critical independent sets of $G_{1}$. Notice that $\mathrm{diadem}%
(G_{1})\subsetneq\mathrm{corona}(G_{1})$.

\begin{theorem}
\label{th7}Let $G$ be a graph whose $\mathrm{core}(G)$ is a critical set. Then

\emph{(i)} $\mathrm{core}(G)\subseteq\mathrm{nucleus}(G)$;

\emph{(ii)} $\mathrm{MaxCritIndep}(G)\vartriangleleft$ $\Omega(G)$;

\emph{(iii)} $\left\vert \mathrm{diadem}(G)\right\vert +\left\vert
\mathrm{nucleus}(G)\right\vert \leq\left\vert \mathrm{corona}(G)\right\vert
+\left\vert \mathrm{core}(G)\right\vert $;

\emph{(iv)} $\mathrm{core}(G)=\mathrm{nucleus}(G)$, if, in addition,
$\mathrm{diadem}(G)=\mathrm{corona}(G)$.
\end{theorem}

\begin{proof}
\emph{(i)} Let $A\in\mathrm{MaxCritIndep}(G)$. According to Theorem \ref{th3},
there exists some $S\in\Omega\left(  G\right)  $, such that $A\subseteq S$.
Since $\mathrm{core}(G)\subseteq S$, it follows that $A\cup\mathrm{core}%
(G)\subseteq S$, and hence $A\cup\mathrm{core}(G)$ is independent. By Theorem
\ref{th4}, we get that $A\cup\mathrm{core}(G)$ is a critical independent set.
Since $A\subseteq A\cup\mathrm{core}(G)$ and $A$ is a maximum critical
independent set, we infer that $\mathrm{core}(G)\subseteq A$. Thus,
$\mathrm{core}(G)\subseteq A$ for every $A\in\mathrm{MaxCritIndep}(G)$.
Therefore, $\mathrm{core}(G)\subseteq\mathrm{nucleus}(G)$.

\emph{(ii)} By Part \emph{(i)}, we know that $\mathrm{core}(G)\subseteq
\mathrm{nucleus}(G)$. According to Theorem \ref{th3}, every critical
independent set is included in some maximum independent set. Hence, we deduce
that $\mathrm{diadem}(G)=\bigcup\mathrm{MaxCritIndep}(G)\subseteq\bigcup
\Omega(G)=\mathrm{corona}(G)$.

\emph{(iii)} The inequality follows from Part \emph{(ii)} and Theorem
\ref{the main theorem}.

\emph{(iv)} Part \emph{(iii)} implies $\left\vert \mathrm{nucleus}%
(G)\right\vert \leq\left\vert \mathrm{core}(G)\right\vert $, and using now
Part \emph{(i)}, we obtain $\mathrm{core}(G)=\mathrm{nucleus}(G)$.
\end{proof}

\begin{corollary}
\label{corollary 4}If $\left\vert \Omega\left(  G\right)  \right\vert \leq2$
and $\mathrm{diadem}(G)=\mathrm{corona}(G)$, then $G$ is a
K\"{o}nig-Egerv\'{a}ry graph.
\end{corollary}

\begin{proof}
If $\left\vert \Omega\left(  G\right)  \right\vert =\left\vert \left\{
S\right\}  \right\vert =1$, then $\mathrm{diadem}(G)=\mathrm{corona}(G)=S$,
and the conclusion follows from Theorem \ref{th5}.

Assume that $\Omega\left(  G\right)  =\left\{  S_{1},S_{2}\right\}  $. Since
$\mathrm{diadem}(G)=\mathrm{corona}(G)$, we infer that the family
$\mathrm{MaxCritIndep}(G)$ contains only two maximum critical independent
sets, say $A_{1}$ and $A_{2}$. By Theorem \ref{th7}\emph{(iv)}, we obtain
$\mathrm{core}(G)=\mathrm{nucleus}(G)$. According to Theorem \ref{th7}, we
have, for instance, $A_{1}\subseteq S_{1}$. Hence, $A_{2}\subseteq S_{2}$,
because, otherwise, $\mathrm{diadem}(G)=A_{1}\cup A_{2}\neq S_{1}\cup
S_{2}=\mathrm{corona}(G)$. If there is some $x\in S_{1}-A_{1}$, then $x\in
A_{2}\subseteq S_{2}$, because $S_{1}-A_{1}\subseteq S_{1}\cup S_{2}=A_{1}\cup
A_{2}$. Therefore, we deduce $x\in S_{1}\cap S_{2}=A_{1}\cap A_{2}$, which
implies $x\in A_{1}$, in contradiction with the assumption that $x\in
S_{1}-A_{1}$. Consequently, $A_{1}=S_{1}$, which ensures, by Theorem
\ref{th5}, that $G$ is a K\"{o}nig-Egerv\'{a}ry graph.
\end{proof}

Theorem \ref{th7}\emph{(i)} holds for every K\"{o}nig-Egerv\'{a}ry graph, with
equality, by Theorem \ref{th5}. The same equality is satisfied by some
non-K\"{o}nig-Egerv\'{a}ry graphs; e.g., the graph $G$ from Figure
\ref{fig51}, where
\begin{align*}
\mathrm{core}(G)  &  =\left\{  v_{1},v_{2},v_{3},v_{6},v_{7},v_{10}\right\}
\cap\left\{  v_{1},v_{2},v_{4},v_{6},v_{7},v_{10}\right\} \\
&  \cap\left\{  v_{1},v_{2},v_{3},v_{6},v_{8},v_{10}\right\}  \cap\left\{
v_{1},v_{2},v_{4},v_{6},v_{8},v_{10}\right\}  .
\end{align*}

\begin{figure}[h]
\setlength{\unitlength}{1cm}\begin{picture}(5,1.9)\thicklines
\multiput(4,0.5)(1,0){5}{\circle*{0.29}}
\multiput(3,1.5)(1,0){4}{\circle*{0.29}}
\multiput(2,0.5)(0,1){2}{\circle*{0.29}}
\put(7,1.5){\circle*{0.29}}
\put(8,1.5){\circle*{0.29}}
\put(2,0.5){\line(1,0){6}}
\put(2,1.5){\line(2,-1){2}}
\put(3,1.5){\line(1,-1){1}}
\put(3,1.5){\line(1,0){1}}
\put(4,0.5){\line(0,1){1}}
\put(4,0.5){\line(1,1){1}}
\put(5,1.5){\line(1,0){1}}
\put(6,0.5){\line(0,1){1}}
\put(8,0.5){\line(0,1){1}}
\put(7,1.5){\line(1,0){1}}
\put(7,1.5){\line(1,-1){1}}
\put(2,0.1){\makebox(0,0){$v_{1}$}}
\put(1.65,1.5){\makebox(0,0){$v_{2}$}}
\put(2.65,1.5){\makebox(0,0){$v_{3}$}}
\put(4.35,1.5){\makebox(0,0){$v_{4}$}}
\put(4,0.1){\makebox(0,0){$v_{5}$}}
\put(5,0.1){\makebox(0,0){$v_{6}$}}
\put(5,1.15){\makebox(0,0){$v_{7}$}}
\put(6,0.1){\makebox(0,0){$v_{9}$}}
\put(6.35,1.5){\makebox(0,0){$v_{8}$}}
\put(7,1.15){\makebox(0,0){$v_{11}$}}
\put(7,0.1){\makebox(0,0){$v_{10}$}}
\put(8,0.1){\makebox(0,0){$v_{12}$}}
\put(8.4,1.5){\makebox(0,0){$v_{13}$}}
\put(1,1){\makebox(0,0){$G$}}
\multiput(10,0.5)(1,0){4}{\circle*{0.29}}
\multiput(11,1.5)(1,0){2}{\circle*{0.29}}
\put(10,0.5){\line(1,0){3}}
\put(11,0.5){\line(0,1){1}}
\put(12,1.5){\line(1,-1){1}}
\put(12,0.5){\line(0,1){1}}
\put(10,0.1){\makebox(0,0){$a$}}
\put(11,0.1){\makebox(0,0){$b$}}
\put(12,0.1){\makebox(0,0){$c$}}
\put(13,0.1){\makebox(0,0){$d$}}
\put(10.65,1.5){\makebox(0,0){$e$}}
\put(11.65,1.5){\makebox(0,0){$f$}}
\put(9.3,1){\makebox(0,0){$H$}}
\end{picture}\caption{\textrm{core}$(G)=\{v_{1},v_{2},v_{6},v_{10}\}$ is a
critical set, since $d\left(  \mathrm{core}(G)\right)  =1=d\left(  G\right)
$.}%
\label{fig51}%
\end{figure}
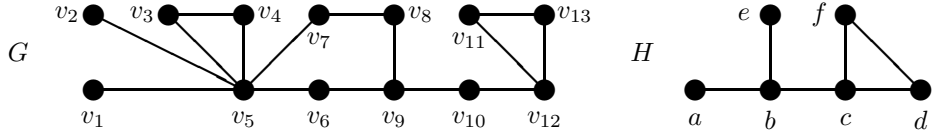

The equality from Theorem \ref{th7}\emph{(iv)} may hold for some graphs where
$\mathrm{diadem}(G)\neq\mathrm{corona}(G)$. For instance, the graph $H$ from
Figure \ref{fig51} satisfies: $\mathrm{core}(H)=\mathrm{nucleus}(H)=\left\{
a,e\right\}  $, $\mathrm{corona}(H)=\left\{  a,e,c,d,f\right\}  $ is a
critical set, but $\mathrm{diadem}(H)=\left\{  a,e\right\}  \neq
\mathrm{corona}(H)$.

\begin{corollary}
If $G$ is a bipartite graph, then $\mathrm{\ker}(G)=\mathrm{core}%
(G)=\mathrm{nucleus}(G)$.
\end{corollary}

The lower bound presented in the following theorem first appeared in
\cite{LevManLemma2011}.

\begin{theorem}
\label{th9} For every graph $G$
\[
2\alpha(G)\leq\left\vert \mathrm{corona}(G)\right\vert +\left\vert
\mathrm{core}(G)\right\vert \leq2\left(  \left\vert V\left(  G\right)
\right\vert -\mu\left(  G\right)  \right)  .
\]

\end{theorem}

\begin{proof}
Let $S\in\Omega(G)$, $\Gamma=\Omega(G)$ and $\Gamma^{\prime}=\{S\}$. As
$\Gamma^{\prime}\subseteq\Gamma$, by Corollary \ref{corollary 1}, we get
\[
2\alpha(G)=2\left\vert S\right\vert \leq\left\vert \mathrm{corona}%
(G)\right\vert +\left\vert \mathrm{core}(G)\right\vert .
\]

For a maximum matching $M$ of $G$, let $A=\left\{  x:\left\{  x,M\left(
x\right)  \right\}  \subseteq\mathrm{corona}(G)\right\}  $, and let $B$
contain all other vertices matched by $M$. Hence, there is no $S\in
\Omega\left(  G\right)  $ such that $x\in$ $S$ and $M\left(  x\right)  \in S$
at the same time. Since $\mathrm{core}(G)\subseteq S\subseteq\mathrm{corona}%
(G)$ for every $S\in\Omega\left(  G\right)  $, we infer that $A\cap
\mathrm{core}(G)=\emptyset$. Thus $A\subseteq\mathrm{corona}(G)-\mathrm{core}%
(G)$, and, consequently,
\[
\left\vert A\right\vert \leq\left\vert \mathrm{corona}(G)\right\vert
-\left\vert \mathrm{core}(G)\right\vert .
\]
On the other hand, for every $x\in B$, we have $1\leq\left\vert \left\{
x,M\left(  x\right)  \right\}  \cap\left(  V\left(  G\right)  -\mathrm{corona}%
(G)\right)  \right\vert $, and this implies%
\[
\left\vert B\right\vert \leq2(\left\vert V\left(  G\right)  \right\vert
-\left\vert \mathrm{corona}(G)\right\vert ).
\]
Consequently, we obtain%
\begin{gather*}
2\mu(G)=2\left\vert M\right\vert =\left\vert A\right\vert +\left\vert
B\right\vert \leq\\
\leq\left\vert \mathrm{corona}(G)\right\vert -\left\vert \mathrm{core}%
(G)\right\vert +2(\left\vert V\left(  G\right)  \right\vert -\left\vert
\mathrm{corona}(G)\right\vert )\\
=2\left\vert V\left(  G\right)  \right\vert -\left\vert \mathrm{corona}%
(G)\right\vert -\left\vert \mathrm{core}(G)\right\vert ,
\end{gather*}
and this completes the proof.
\end{proof}

\begin{corollary}
If $\emptyset\neq\Gamma\subseteq\Omega(G)$, then
\[
2\alpha(G)\leq\left\vert
{\displaystyle\bigcup}
\Gamma\right\vert +\left\vert
{\displaystyle\bigcap}
\Gamma\right\vert \leq2\left(  \left\vert V\left(  G\right)  \right\vert
-\mu(G)\right)  .
\]

\end{corollary}

\begin{proof}
Since $\emptyset\neq\Gamma\subseteq\Omega(G)$, we have $\Gamma\vartriangleleft
\Omega(G)$. Combining Corollary \ref{corollary 8}, Corollary \ref{corollary 1}
and Theorem \ref{th9}, we infer that
\[
2\alpha(G)\leq\left\vert
{\displaystyle\bigcup}
\Gamma\right\vert +\left\vert
{\displaystyle\bigcap}
\Gamma\right\vert \leq\left\vert \mathrm{corona}(G)\right\vert +\left\vert
\mathrm{core}(G)\right\vert \leq2\left(  \left\vert V\left(  G\right)
\right\vert -\mu\left(  G\right)  \right)  ,
\]
as claimed.
\end{proof}

Clearly, $G$ is a K\"{o}nig-Egerv\'{a}ry graph if and only if the lower and
upper bounds in Theorem \ref{th9} coincide.

The graphs from Figure \ref{fig11} satisfy:%
\begin{align*}
2\alpha(G_{1})  &  =4=\left\vert \mathrm{corona}(G_{1})\right\vert +\left\vert
\mathrm{core}(G_{1})\right\vert <2\left(  \left\vert V\left(  G_{1}\right)
\right\vert -\mu(G_{1})\right)  =6,\\
2\alpha(G_{2})  &  =6<\left\vert \mathrm{corona}(G_{2})\right\vert +\left\vert
\mathrm{core}(G_{2})\right\vert =8=2\left(  \left\vert V\left(  G_{2}\right)
\right\vert -\mu(G_{2})\right) \\
2\alpha(G_{3})  &  =12<\left\vert \mathrm{corona}(G_{3})\right\vert
+\left\vert \mathrm{core}(G_{3})\right\vert =13<2\left(  \left\vert V\left(
G_{1}\right)  \right\vert -\mu(G_{1})\right)  =14,
\end{align*}
i.e., the bounds from Theorem \ref{th9} are tight.

\begin{figure}[h]
\setlength{\unitlength}{1cm}\begin{picture}(5,1.5)\thicklines
\multiput(0.5,0)(1,0){3}{\circle*{0.29}}
\multiput(0.5,1)(2,0){2}{\circle*{0.29}}
\put(0.5,0){\line(1,0){2}}
\put(0.5,0){\line(0,1){1}}
\put(0.5,1){\line(1,-1){1}}
\put(1.5,0){\line(1,1){1}}
\put(2.5,0){\line(0,1){1}}
\put(1.5,1.4){\makebox(0,0){$G_{1}$}}
\multiput(3.5,0)(1,0){4}{\circle*{0.29}}
\multiput(3.5,1)(1,0){4}{\circle*{0.29}}
\put(3.5,0){\line(1,0){3}}
\put(3.5,0){\line(0,1){1}}
\put(3.5,1){\line(1,-1){1}}
\put(3.5,0){\line(1,1){1}}
\put(4.5,0){\line(0,1){1}}
\put(5.5,1){\line(1,0){1}}
\put(5.5,0){\line(0,1){1}}
\put(5.5,0){\line(1,1){1}}
\put(5.5,1){\line(1,-1){1}}
\put(6.5,0){\line(0,1){1}}
\put(5,1.4){\makebox(0,0){$G_{2}$}}
\multiput(7.5,0)(1,0){6}{\circle*{0.29}}
\multiput(7.5,1)(1,0){6}{\circle*{0.29}}
\put(7.5,0){\line(1,0){5}}
\put(7.5,0){\line(0,1){1}}
\put(7.5,0){\line(1,1){1}}
\put(8.5,0){\line(1,1){1}}
\put(9.5,1){\line(1,0){1}}
\put(10.5,0){\line(0,1){1}}
\put(11.5,0){\line(0,1){1}}
\put(12.5,0){\line(0,1){1}}
\put(10,1.4){\makebox(0,0){$G_{3}$}}
\end{picture}\caption{$G_{1},G_{2}$ and $G_{3}$ are non-K\"{o}nig-Egerv\'{a}ry
graphs.}%
\label{fig11}%
\end{figure}

\begin{remark}
For each $n\geq1$, the graph $K_{2n}$ satisfies $\left\vert \mathrm{corona}%
(K_{2n})\right\vert +\left\vert \mathrm{core}(K_{2n})\right\vert =2n=2\left(
\left\vert V\left(  K_{2n}\right)  \right\vert -\mu(K_{2n})\right)  $.
\end{remark}

\begin{remark}
Let $G$ be the graph obtained by joining to pendant vertices to one of the
vertices of $K_{2n+1}$. Then $\left\vert \mathrm{corona}(G)\right\vert
+\left\vert \mathrm{core}(G)\right\vert =2+2n=2\left(  \left\vert G\right\vert
-\mu(G)\right)  $.
\end{remark}

The graphs from Figure \ref{fig111} satisfy:%
\begin{align*}
2\alpha(G_{1})  &  =4=\left\vert \mathrm{corona}(G_{1})\right\vert +\left\vert
\mathrm{core}(G_{1})\right\vert <2\left(  \left\vert V\left(  G_{1}\right)
\right\vert -\mu(G_{1})\right)  =6,\\
2\alpha(G_{2})  &  =6<\left\vert \mathrm{corona}(G_{2})\right\vert +\left\vert
\mathrm{core}(G_{2})\right\vert =8=2\left(  \left\vert V\left(  G_{2}\right)
\right\vert -\mu(G_{2})\right) \\
2\alpha(G_{3})  &  =8<\left\vert \mathrm{corona}(G_{3})\right\vert +\left\vert
\mathrm{core}(G_{3})\right\vert =9<2\left(  \left\vert V\left(  G_{1}\right)
\right\vert -\mu(G_{1})\right)  =11,
\end{align*}
i.e., the bounds from Theorem \ref{th9} are tight.

\begin{figure}[h]
\setlength{\unitlength}{1cm}\begin{picture}(5,1)\thicklines
\multiput(1.5,0)(1,0){3}{\circle*{0.29}}
\multiput(1.5,1)(2,0){2}{\circle*{0.29}}
\put(1.5,0){\line(1,0){2}}
\put(1.5,0){\line(0,1){1}}
\put(1.5,1){\line(1,-1){1}}
\put(2.5,0){\line(1,1){1}}
\put(3.5,0){\line(0,1){1}}
\put(0.5,0.5){\makebox(0,0){$G_{1}$}}
\multiput(5.5,0)(1,0){4}{\circle*{0.29}}
\multiput(5.5,1)(1,0){4}{\circle*{0.29}}
\put(5.5,0){\line(1,0){3}}
\put(5.5,0){\line(0,1){1}}
\put(5.5,1){\line(1,-1){1}}
\put(5.5,0){\line(1,1){1}}
\put(6.5,0){\line(0,1){1}}
\put(7.5,1){\line(1,0){1}}
\put(7.5,0){\line(0,1){1}}
\put(7.5,0){\line(1,1){1}}
\put(7.5,1){\line(1,-1){1}}
\put(8.5,0){\line(0,1){1}}
\put(4.5,0.5){\makebox(0,0){$G_{2}$}}
\multiput(10.5,0)(1,0){4}{\circle*{0.29}}
\multiput(10.5,1)(1,0){4}{\circle*{0.29}}
\put(10.5,0){\line(1,0){3}}
\put(10.5,0){\line(0,1){1}}
\put(10.5,0){\line(1,1){1}}
\put(11.5,0){\line(1,1){1}}
\put(12.5,1){\line(1,0){1}}
\put(13.5,0){\line(0,1){1}}
\put(9.5,0.5){\makebox(0,0){$G_{3}$}}
\end{picture}\caption{$G_{1},G_{2}$ and $G_{3}$ are non-K\"{o}nig-Egerv\'{a}ry
graphs.}%
\label{fig111}%
\end{figure}

\begin{corollary}
\label{the fact about KE}\cite{LevManLemma2011} If $G$ is a
K\"{o}nig-Egerv\'{a}ry graph, then $\left\vert \mathrm{corona}(G)\right\vert
+\left\vert \mathrm{core}(G)\right\vert =2\alpha(G)$.
\end{corollary}

\begin{proof}
Since $G$ is a K\"{o}nig-Egerv\'{a}ry graph, we have $\alpha(G)=\left\vert
V\left(  G\right)  \right\vert -\mu(G)$, and according to Theorem \ref{th9},
we get
\[
2\alpha(G)\leq\left\vert \mathrm{corona}(G)\right\vert +\left\vert
\mathrm{core}(G)\right\vert \leq2\left(  \left\vert V\left(  G\right)
\right\vert -\mu(G)\right)  =2\alpha(G),
\]
and this completes the proof.
\end{proof}

It is known that $\left\vert V\left(  G\right)  \right\vert -1\leq
\alpha(G)+\mu\left(  G\right)  \leq\left\vert V\left(  G\right)  \right\vert $
for every unicyclic graph \cite{LevMan2012c}.

\begin{theorem}
\label{th10}\cite{LevMan2014} If $G$ is a unicyclic graph, then
\[
2\alpha(G)\leq\left\vert \mathrm{corona}(G)\right\vert +\left\vert
\mathrm{core}(G)\right\vert \leq2\alpha(G)+1.
\]

\end{theorem}

By Corollary \ref{the fact about KE} and Theorem \ref{th10} we know that every
unicyclic non-K\"{o}nig-Egerv\'{a}ry graph satisfies the equalities
$\left\vert V\left(  G\right)  \right\vert -1=\alpha(G)+\mu\left(  G\right)  $
and $\left\vert \mathrm{corona}(G)\right\vert +\left\vert \mathrm{core}%
(G)\right\vert =2\alpha(G)+1$. Consequently,
\[
2\alpha(G)+1=2\left(  \left\vert V\left(  G\right)  \right\vert -\mu\left(
G\right)  \right)  -1<2\left(  \left\vert V\left(  G\right)  \right\vert
-\mu(G)\right)  ,
\]
which improves on the upper bound in Theorem \ref{th9}, in the case of
unicyclic graphs.

\begin{corollary}
If $G$ is a K\"{o}nig-Egerv\'{a}ry graph, then $\left\vert \mathrm{diadem}%
(G)\right\vert +\left\vert \mathrm{nucleus}(G)\right\vert =2\alpha(G)$.
\end{corollary}

\begin{proof}
By Theorem \ref{th5}, we have that $\mathrm{diadem}(G)=\mathrm{corona}(G)$ and
$\mathrm{core}(G)$ is critical. Combining Theorem \ref{th7}\emph{(iv)} and
Corollary \ref{the fact about KE}, we get the result.
\end{proof}

\begin{corollary}
\label{2alpha in KE} If $G$ is a K\"{o}nig-Egerv\'{a}ry graph, and
$\emptyset\neq\Gamma\subseteq\Omega\left(  G\right)  $, then
\[
\left\vert \bigcup\Gamma\right\vert +\left\vert \bigcap\Gamma\right\vert
=2\alpha(G).
\]

\end{corollary}

\begin{proof}
Let $S\in\Gamma$ and define $\Gamma^{\prime}=\{S\}$. Hence, we have
$\Gamma^{\prime}\vartriangleleft\Gamma\vartriangleleft\Omega\left(  G\right)
$. By Theorem \ref{the main theorem} and Corollary \ref{the fact about KE}, we
obtain
\[
2\alpha(G)=f\left(  \Gamma^{\prime}\right)  \leq f\left(  \Gamma\right)
=\left\vert \bigcup\Gamma\right\vert +\left\vert \bigcap\Gamma\right\vert
\leq\left\vert \mathrm{corona}(G)\right\vert +\left\vert \mathrm{core}%
(G)\right\vert =2\alpha(G),
\]
which clearly implies $\left\vert \bigcup\Gamma\right\vert +\left\vert
\bigcap\Gamma\right\vert =2\alpha(G)$.
\end{proof}

Let us notice that the converse of Corollary \ref{2alpha in KE} is not
necessarily true. For instance, the graphs $G_{1}$ and $G_{2}$ from Figure
\ref{fig244}, clearly, both satisfy: $\left\vert \bigcup\Gamma\right\vert
+\left\vert \bigcap\Gamma\right\vert =2\alpha(G)$ for every $\emptyset
\neq\Gamma\subseteq\Omega\left(  G\right)  $, but none is a
K\"{o}nig-Egerv\'{a}ry graph.\begin{figure}[h]
\setlength{\unitlength}{1cm}\begin{picture}(5,1.7)\thicklines
\multiput(4,0.5)(1,0){4}{\circle*{0.29}}
\multiput(4,1.5)(2,0){2}{\circle*{0.29}}
\put(7,1.5){\circle*{0.29}}
\put(4,0.5){\line(1,0){3}}
\put(4,0.5){\line(0,1){1}}
\put(4,1.5){\line(1,-1){1}}
\put(4,1.5){\line(2,-1){2}}
\put(4,0.5){\line(2,1){2}}
\put(5,0.5){\line(1,1){1}}
\put(6,0.5){\line(0,1){1}}
\put(7,0.5){\line(0,1){1}}
\qbezier(4,0.5)(5,-0.3)(6,0.5)
\put(3.7,1.5){\makebox(0,0){$a$}}
\put(6.3,1.5){\makebox(0,0){$b$}}
\put(7.3,1.5){\makebox(0,0){$c$}}
\put(7.3,0.5){\makebox(0,0){$d$}}
\put(3,1){\makebox(0,0){$G_{1}$}}
\multiput(9.5,0.5)(1,0){3}{\circle*{0.29}}
\multiput(10.5,1.5)(1,0){2}{\circle*{0.29}}
\put(9.5,0.5){\line(1,0){2}}
\put(9.5,0.5){\line(1,1){1}}
\put(9.5,0.5){\line(2,1){2}}
\put(10.5,0.5){\line(0,1){1}}
\put(10.5,0.5){\line(1,1){1}}
\put(10.5,1.5){\line(1,-1){1}}
\put(10.5,1.5){\line(1,0){1}}
\put(11.5,0.5){\line(0,1){1}}
\put(9.15,0.5){\makebox(0,0){$u$}}
\put(11.8,0.5){\makebox(0,0){$v$}}
\put(8.5,1){\makebox(0,0){$G_{2}$}}
\end{picture}\caption{$\Omega\left(  G_{1}\right)  =\left\{  \left\{
a,b,c\right\}  ,\left\{  a,b,d\right\}  \right\}  $, while $\Omega\left(
G_{2}\right)  =\left\{  u,v\right\}  $.}%
\label{fig244}%
\end{figure}

\section{A characterization of K\"{o}nig-Egerv\'{a}ry graphs}

\begin{theorem}
\label{perfect matching}If $\Gamma\subseteq\Omega\left(  G\right)  $ and
$\left\vert
{\displaystyle\bigcup}
\Gamma\right\vert +\left\vert
{\displaystyle\bigcap}
\Gamma\right\vert =2\alpha(G)$, then

\emph{(i)} there is a perfect matching in $G\left[
{\displaystyle\bigcup}
\Gamma-%
{\displaystyle\bigcap}
\Gamma\right]  $;

\emph{(ii) }$\left\vert
{\displaystyle\bigcup}
\Gamma\right\vert -\left\vert
{\displaystyle\bigcap}
\Gamma\right\vert =2\mu\left(  G\left[
{\displaystyle\bigcup}
\Gamma\right]  \right)  $;

\emph{(iii)} $\alpha(G\left[
{\displaystyle\bigcup}
\Gamma\right]  )=\alpha(G)$;

\emph{(iv) }$G\left[
{\displaystyle\bigcup}
\Gamma\right]  $ is a K\"{o}nig-Egerv\'{a}ry graph.
\end{theorem}

\begin{proof}
\emph{(i) }Let $S\in\Gamma$. By Lemma \ref{MatchingLemma}, there is a
matching, say $M$, from $S-%
{\displaystyle\bigcap}
\Gamma$ into $%
{\displaystyle\bigcup}
\Gamma-S$. On the other hand,
\begin{gather*}
\left\vert S-\bigcap\Gamma\right\vert =\left\vert S\right\vert -\left\vert
\bigcap\Gamma\right\vert =\alpha(G)-\left\vert \bigcap\Gamma\right\vert \\
=\left\vert \bigcup\Gamma\right\vert -\alpha(G)=\left\vert \bigcup
\Gamma\right\vert -\left\vert S\right\vert =\left\vert \bigcup\Gamma
-S\right\vert .
\end{gather*}
Since
\[
\left(  S-%
{\displaystyle\bigcap}
\Gamma\right)  \cup\left(
{\displaystyle\bigcup}
\Gamma-S\right)  =%
{\displaystyle\bigcup}
\Gamma-%
{\displaystyle\bigcap}
\Gamma,
\]
we conclude that $M$\ is a perfect matching in $G\left[
{\displaystyle\bigcup}
\Gamma-%
{\displaystyle\bigcap}
\Gamma\right]  $.

\emph{(ii) }By Part \emph{(i)}, there is a perfect matching in $G\left[
{\displaystyle\bigcup}
\Gamma-%
{\displaystyle\bigcap}
\Gamma\right]  $. Hence,
\[
2\mu\left(  G\left[
{\displaystyle\bigcup}
\Gamma\right]  \right)  \geq2\mu\left(  G\left[  \left\vert
{\displaystyle\bigcup}
\Gamma\right\vert -\left\vert
{\displaystyle\bigcap}
\Gamma\right\vert \right]  \right)  =\left\vert
{\displaystyle\bigcup}
\Gamma\right\vert -\left\vert
{\displaystyle\bigcap}
\Gamma\right\vert .
\]
It remains to prove that
\[
2\mu\left(  G\left[
{\displaystyle\bigcup}
\Gamma\right]  \right)  \leq\left\vert
{\displaystyle\bigcup}
\Gamma\right\vert -\left\vert
{\displaystyle\bigcap}
\Gamma\right\vert .
\]
Let $M$ be a maximum matching in $G\left[
{\displaystyle\bigcup}
\Gamma\right]  $. Since all the members of $\Gamma$ are independent sets,
there exists no edge $xy$ such that $x\in%
{\displaystyle\bigcap}
\Gamma$ and $y\in%
{\displaystyle\bigcup}
\Gamma$. Therefore, $V\left(  M\right)  \cap%
{\displaystyle\bigcap}
\Gamma=\emptyset$. Finally, we get
\[
2\mu\left(  G\left[
{\displaystyle\bigcup}
\Gamma\right]  \right)  =\left\vert V\left(  M\right)  \right\vert
\leq\left\vert
{\displaystyle\bigcup}
\Gamma\right\vert -\left\vert
{\displaystyle\bigcap}
\Gamma\right\vert .
\]

\emph{(iii) }On the one hand, $\alpha(G\left[
{\displaystyle\bigcup}
\Gamma\right]  )\leq\alpha(G)$, because every independent set in $G\left[
{\displaystyle\bigcup}
\Gamma\right]  $ is independent in $G$ as well. On the other hand, if
$S\in\Gamma$, then $\left\vert S\right\vert =\alpha(G)$, and $S$ is
independent in $G\left[
{\displaystyle\bigcup}
\Gamma\right]  $. Thus $\alpha(G\left[
{\displaystyle\bigcup}
\Gamma\right]  )\geq\alpha(G)$, and consequently, we obtain $\alpha(G\left[
{\displaystyle\bigcup}
\Gamma\right]  )=\alpha(G)$.

\emph{(iv) }Using the hypothesis and Part \emph{(ii)}, we deduce that%
\[
2\alpha(G)-\left\vert
{\displaystyle\bigcap}
\Gamma\right\vert =\left\vert
{\displaystyle\bigcup}
\Gamma\right\vert =\left\vert
{\displaystyle\bigcap}
\Gamma\right\vert +2\mu\left(  G\left[
{\displaystyle\bigcup}
\Gamma\right]  \right)  ,
\]
which, by our assumption and Part \emph{(iii)}, implies
\[
2\left\vert
{\displaystyle\bigcup}
\Gamma\right\vert =2\mu\left(  G\left[
{\displaystyle\bigcup}
\Gamma\right]  \right)  +2\alpha(G)=2\mu\left(  G\left[
{\displaystyle\bigcup}
\Gamma\right]  \right)  +2\alpha(G\left[
{\displaystyle\bigcup}
\Gamma\right]  ),
\]
i.e., $G\left[
{\displaystyle\bigcup}
\Gamma\right]  $ is a K\"{o}nig-Egerv\'{a}ry graph.
\end{proof}

In particular, if we take $\Gamma=\Omega\left(  G\right)  $ in Theorem
\ref{perfect matching}, we get the following.

\begin{corollary}
\label{cor}If $\left\vert \mathrm{corona}(G)\right\vert +\left\vert
\mathrm{core}(G)\right\vert =2\alpha(G)$, then $G\left[  \mathrm{corona}%
(G)\right]  $ is a K\"{o}nig-Egerv\'{a}ry graph.
\end{corollary}

Notice that the equality $\left\vert \mathrm{corona}(G)\right\vert +\left\vert
\mathrm{core}(G)\right\vert =2\alpha(G)$ is not enough to infer that $G$ is a
K\"{o}nig-Egerv\'{a}ry graph, e.g., see the graph $G_{1}$ from Figure
\ref{fig233}, that has: $\alpha\left(  G_{1}\right)  =3$, \textrm{core}%
$(G_{1})=\{d\}$, \textrm{corona}$(G_{1})=\{a,b,d,f,g\}$,.

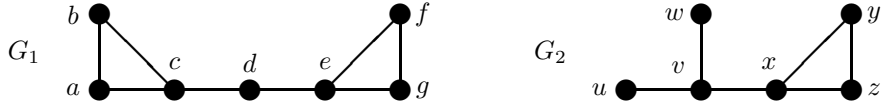
\begin{figure}[h]
\setlength{\unitlength}{1.0cm} \begin{picture}(5,1.5)\thicklines
\multiput(2,0)(1,0){5}{\circle*{0.29}}
\multiput(2,1)(4,0){2}{\circle*{0.29}}
\put(2,0){\line(1,0){4}}
\put(2,1){\line(1,-1){1}}
\put(2,0){\line(0,1){1}}
\put(5,0){\line(1,1){1}}
\put(6,0){\line(0,1){1}}
\put(1.65,0){\makebox(0,0){$a$}}
\put(1.65,1){\makebox(0,0){$b$}}
\put(3,0.35){\makebox(0,0){$c$}}
\put(4,0.35){\makebox(0,0){$d$}}
\put(5,0.35){\makebox(0,0){$e$}}
\put(6.3,1){\makebox(0,0){$f$}}
\put(6.3,0){\makebox(0,0){$g$}}
\put(1,0.5){\makebox(0,0){$G_{1}$}}
\multiput(9,0)(1,0){4}{\circle*{0.29}}
\multiput(10,1)(2,0){2}{\circle*{0.29}}
\put(9,0){\line(1,0){3}}
\put(10,0){\line(0,1){1}}
\put(11,0){\line(1,1){1}}
\put(12,0){\line(0,1){1}}
\put(8.65,0){\makebox(0,0){$u$}}
\put(9.65,1){\makebox(0,0){$w$}}
\put(9.7,0.3){\makebox(0,0){$v$}}
\put(10.9,0.35){\makebox(0,0){$x$}}
\put(12.3,1){\makebox(0,0){$y$}}
\put(12.3,0){\makebox(0,0){$z$}}
\put(8,0.5){\makebox(0,0){$G_{2}$}}
\end{picture}\caption{$\alpha\left(  G_{1}\right)  =3$, \textrm{core}%
$(G_{1})=\{d\}$, \textrm{corona}$(G_{1})=\{a,b,d,f,g\}$, while \textrm{core}%
$(G_{2})=\{u,w\}$ and $V\left(  G_{2}\right)  -$ \textrm{corona}%
$(G_{2})=\{v\}$.}%
\label{fig233}%
\end{figure}

\begin{corollary}
If $G$ is a K\"{o}nig-Egerv\'{a}ry graph, and $\emptyset\neq$ $\Gamma
\subseteq\Omega(G)$, then

\emph{(i)} $\alpha(G\left[
{\displaystyle\bigcup}
\Gamma\right]  )=\alpha(G)$ and $\Gamma\subseteq\Omega(G\left[
{\displaystyle\bigcup}
\Gamma\right]  )$;

\emph{(ii)} $\mathrm{corona}(G\left[
{\displaystyle\bigcup}
\Gamma\right]  )=%
{\displaystyle\bigcup}
\Gamma$ and $\mathrm{core}(G\left[
{\displaystyle\bigcup}
\Gamma\right]  )=%
{\displaystyle\bigcap}
\Gamma$.
\end{corollary}

\begin{proof}
\emph{(i)} It is true by Corollary \ref{2alpha in KE} and Theorem
\ref{perfect matching}\emph{(iii)}.

\emph{(ii) }Since $V\left(  G\left[
{\displaystyle\bigcup}
\Gamma\right]  \right)  =%
{\displaystyle\bigcup}
\Gamma$, we have $\mathrm{corona}(G\left[
{\displaystyle\bigcup}
\Gamma\right]  )\subseteq%
{\displaystyle\bigcup}
\Gamma$. But by Part \emph{(i)}, $%
{\displaystyle\bigcup}
\Gamma\subseteq\mathrm{corona}(G\left[
{\displaystyle\bigcup}
\Gamma\right]  )$.

By Part \emph{(i)}, $\mathrm{core}(G\left[
{\displaystyle\bigcup}
\Gamma\right]  )=%
{\displaystyle\bigcap}
\Omega(G\left[
{\displaystyle\bigcup}
\Gamma\right]  )\subseteq%
{\displaystyle\bigcap}
\Gamma$, so it is enough to prove that $\left\vert \mathrm{core}(G\left[
{\displaystyle\bigcup}
\Gamma\right]  )\right\vert =\left\vert
{\displaystyle\bigcap}
\Gamma\right\vert $. According to Theorem \ref{perfect matching}\emph{(iv)},
$G\left[
{\displaystyle\bigcup}
\Gamma\right]  $ is K\"{o}nig-Egerv\'{a}ry. Therefore, using Corollary
\ref{the fact about KE}, we get%
\begin{align*}
2\alpha\left(  G\left[
{\displaystyle\bigcup}
\Gamma\right]  \right)   &  =\left\vert \mathrm{corona}(G\left[
{\displaystyle\bigcup}
\Gamma\right]  )\right\vert +\left\vert \mathrm{core}(G\left[
{\displaystyle\bigcup}
\Gamma\right]  )\right\vert \\
&  =\left\vert
{\displaystyle\bigcup}
\Gamma\right\vert +\left\vert \mathrm{core}(G\left[
{\displaystyle\bigcup}
\Gamma\right]  )\right\vert .
\end{align*}
Since, by Corollary \ref{2alpha in KE}, we have that $\left\vert
{\displaystyle\bigcup}
\Gamma\right\vert +\left\vert
{\displaystyle\bigcap}
\Gamma\right\vert =2\alpha(G)$, we finally obtain the equality $\left\vert
\mathrm{core}(G\left[
{\displaystyle\bigcup}
\Gamma\right]  )\right\vert =\left\vert
{\displaystyle\bigcap}
\Gamma\right\vert $, as claimed.
\end{proof}

The following proposition shows that a characterization of
K\"{o}nig-Egerv\'{a}ry graphs cannot relate only to the maximum independent sets.

\begin{proposition}
For every K\"{o}nig-Egerv\'{a}ry graph $G\notin\left\{  K_{1},K_{2}\right\}
$, there is a non-K\"{o}nig-Egerv\'{a}ry graph $G^{\prime}$, such that $G$ is
an induced subgraph of $G^{\prime}$ and $\Omega(G^{\prime})=\Omega(G)$.
\end{proposition}

\begin{proof}
Let $n=\left\vert V(G)\right\vert $, and $K_{n+1}$ be a complete graph, such
that $V\left(  G\right)  \cap V\left(  K_{n+1}\right)  =\emptyset$. We define
$G^{\prime}$ as the graph having:%
\begin{align*}
V\left(  G^{\prime}\right)   &  =V\left(  G\right)  \cup V\left(
K_{n+1}\right) \\
E\left(  G^{\prime}\right)   &  =E\left(  G\right)  \cup E\left(
K_{n+1}\right)  \cup\left\{  xy:x\in V\left(  G\right)  ,y\in V\left(
K_{n+1}\right)  \right\}  .
\end{align*}
Clearly, $\alpha(G^{\prime})=\alpha(G)$, $\left\vert V(G^{\prime})\right\vert
=2n+1$, and $\mu(G^{\prime})=n$. Hence we get
\[
\alpha(G^{\prime})+\mu(G^{\prime})=\alpha(G)+n<2n+1=\left\vert V(G^{\prime
})\right\vert .
\]
Therefore, $G^{\prime}$ is not a K\"{o}nig-Egerv\'{a}ry graph, while $G$ is an
induced subgraph of $G^{\prime}$ that clearly satisfies $\Omega(G^{\prime
})=\Omega(G)$.
\end{proof}

Let us mention that the difference $\left\vert V\left(  G\right)
-\mathrm{corona}(G)\right\vert -\left\vert \mathrm{core}(G)\right\vert $ may
reach any positive integer. For instance, $G=K_{n}-e,n\geq4$.

\begin{theorem}
\label{characterization of ke graph1}For a graph $G$ the following assertions
are equivalent:

\emph{(i)} $G$ is a K\"{o}nig-Egerv\'{a}ry graph;

\emph{(ii)} for every $S_{1},S_{2}\in\Omega(G)$ there is a matching from
$V\left(  G\right)  -S_{1}\cup S_{2}$ into $S_{1}\cap S_{2}$;

\emph{(iii)} there exist $S_{1},S_{2}\in\Omega(G)$, such that there is a
matching from $V\left(  G\right)  -S_{1}\cup S_{2}$ into $S_{1}\cap S_{2}$.
\end{theorem}

\begin{proof}
\emph{(i) }$\Rightarrow$\emph{\ (ii) }Suppose that $G$ is a
K\"{o}nig-Egerv\'{a}ry graph, and let $H=G\left[  S_{1}\cup S_{2}\right]  $
for some arbitrary $S_{1},S_{2}\in\Omega(G)$. Since $\left\vert S_{1}\cup
S_{2}\right\vert +\left\vert S_{1}\cap S_{2}\right\vert =2\alpha(G)$, Theorem
\ref{perfect matching}\emph{(ii),(iv)} ensures that $\alpha\left(  G\right)
=\alpha(H)$ and $H$ is also a K\"{o}nig-Egerv\'{a}ry graph. Thus
\[
\left\vert V\left(  G\right)  \right\vert -\mu\left(  G\right)  =\alpha\left(
G\right)  =\alpha(H)=\left\vert V\left(  H\right)  \right\vert -\mu(H),
\]
and, consequently,
\[
\left\vert V\left(  G\right)  \right\vert -\left\vert V\left(  H\right)
\right\vert =\mu\left(  G\right)  -\mu(H).
\]
Let $M$ be a maximum matching of $G$. Applying Theorem \ref{perfect matching}%
\emph{(ii)} with $\Gamma=\left\{  S_{1},S_{2}\right\}  $, we obtain%
\[
\left\vert S_{1}\cup S_{2}\right\vert -\left\vert S_{1}\cap S_{2}\right\vert
=2\mu(H).
\]
Hence,
\begin{gather*}
\left\vert V\left(  M\right)  -S_{1}\cap S_{2}\right\vert \leq\left\vert
V\left(  G\right)  -S_{1}\cap S_{2}\right\vert =\left\vert V\left(  G\right)
-S_{1}\cup S_{2}\right\vert +\left\vert S_{1}\cup S_{2}-S_{1}\cap
S_{2}\right\vert \\
=(\left\vert V\left(  G\right)  \right\vert -\left\vert V\left(  H\right)
\right\vert )+2\mu(H)=(\mu\left(  G\right)  -\mu(H))+2\mu(H)=\mu\left(
G\right)  +\mu(H).
\end{gather*}
Therefore, $\left\vert V\left(  M\right)  -S_{1}\cap S_{2}\right\vert \leq
\mu\left(  G\right)  +\mu(H)$. But
\[
2\mu\left(  G\right)  =\left\vert V\left(  M\right)  \right\vert =\left\vert
V\left(  M\right)  \cap S_{1}\cap S_{2}\right\vert +\left\vert V\left(
M\right)  -S_{1}\cap S_{2}\right\vert .
\]
Thus
\begin{gather*}
\left\vert V\left(  M\right)  \cap S_{1}\cap S_{2}\right\vert \geq2\mu\left(
G\right)  -(\mu\left(  G\right)  +\mu(H))=\\
=\mu\left(  G\right)  -\mu(H)=\left\vert V\left(  G\right)  \right\vert
-\left\vert V\left(  H\right)  \right\vert =|V\left(  G\right)  -S_{1}\cup
S_{2}|.
\end{gather*}
Clearly, $M(y)\in V\left(  G\right)  -S_{1}\cup S_{2}$ for every $y\in
V\left(  M\right)  \cap S_{1}\cap S_{2}$. In other words, $M$ induces an
injective mapping, say $M_{1}$, from $V\left(  M\right)  \cap S_{1}\cap S_{2}$
into $V\left(  G\right)  -S_{1}\cup S_{2}$. Since $\left\vert V\left(
M\right)  \cap S_{1}\cap S_{2}\right\vert \geq|V\left(  G\right)  -S_{1}\cup
S_{2}|$, we conclude that $M_{1}$ is a bijection. Therefore, $M_{1}^{-1}$ is a
matching from $V\left(  G\right)  -S_{1}\cup S_{2}$ into $V\left(  M\right)
\cap S_{1}\cap S_{2}$. Hence, $M_{1}^{-1}$ is a matching from $V\left(
G\right)  -S_{1}\cup S_{2}$ into $S_{1}\cap S_{2}$.

\emph{(ii) }$\Rightarrow$\emph{\ (iii) }It is clear.

\emph{(iii) }$\Rightarrow$\emph{\ (i) }Suppose\emph{ }that\emph{ }there exist
two sets $S_{1},S_{2}\in\Omega(G)$, such that there is a matching, say $M_{1}%
$, from $V\left(  G\right)  -S_{1}\cup S_{2}$ into $S_{1}\cap S_{2}$. Let
$H=G\left[  S_{1}\cup S_{2}\right]  $. In general, $\mu(G)\leq\mu(G-v)+1$ for
every vertex $v\in V\left(  G\right)  $. Consequently,$\ $we have $\left\vert
V\left(  G\right)  \right\vert -\left\vert V\left(  H\right)  \right\vert
\geq\mu\left(  G\right)  -\mu(H)$, because $H$ is a subgraph of $G$

Let $M_{2}$ be a maximum matching in $G\left[  S_{1}\cup S_{2}-S_{1}\cap
S_{2}\right]  $. Since there are no edges connecting $S_{1}\cup S_{2}%
-S_{1}\cap S_{2}$ and $S_{1}\cap S_{2}$, we infer that $M_{1}\cup M_{2}$ is a
matching in $G$. Consequently, we obtain
\begin{gather*}
\mu\left(  G\right)  \geq\left\vert M_{1}\right\vert +\left\vert
M_{2}\right\vert =\left\vert V\left(  G\right)  -S_{1}\cup S_{2}\right\vert
+\mu(G\left[  S_{1}\cup S_{2}-S_{1}\cap S_{2}\right]  )=\\
=\left\vert V\left(  G\right)  \right\vert -\left\vert S_{1}\cup
S_{2}\right\vert +\mu(G\left[  S_{1}\cup S_{2}\right]  )=\left\vert V\left(
G\right)  \right\vert -\left\vert V\left(  H\right)  \right\vert +\mu(H).
\end{gather*}
Hence, $\left\vert V\left(  G\right)  \right\vert -\left\vert V\left(
H\right)  \right\vert =\mu\left(  G\right)  -\mu(H)$.

By Theorem \ref{perfect matching}\emph{(iii),(iv)}, we infer that $H$ is a
K\"{o}nig-Egerv\'{a}ry graph, and $\alpha\left(  H\right)  =\alpha\left(
G\right)  $. Therefore,%
\[
\left\vert V\left(  G\right)  \right\vert -\mu\left(  G\right)  =\left\vert
V\left(  H\right)  \right\vert -\mu(H)=\alpha\left(  H\right)  =\alpha\left(
G\right)  .
\]
Thus $\left\vert V\left(  G\right)  \right\vert =\alpha\left(  G\right)
+\mu(G)$, which means that $G$ is a K\"{o}nig-Egerv\'{a}ry graph as well.
\end{proof}

The conditions \emph{(ii)} or \emph{(iii)} from Theorem
\ref{characterization of ke graph1} are not equivalent when we take more than
two maximum independent sets. 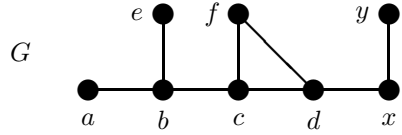
\begin{figure}[h]
\setlength{\unitlength}{1cm}\begin{picture}(5,1.9)\thicklines
\multiput(5,0.5)(1,0){4}{\circle*{0.29}}
\multiput(6,1.5)(1,0){2}{\circle*{0.29}}
\put(5,0.5){\line(1,0){3}}
\put(6,0.5){\line(0,1){1}}
\put(7,1.5){\line(1,-1){1}}
\put(7,0.5){\line(0,1){1}}
\multiput(9,0.5)(0,1){2}{\circle*{0.29}}
\put(8,0.5){\line(1,0){1}}
\put(9,0.5){\line(0,1){1}}
\put(5,0.1){\makebox(0,0){$a$}}
\put(6,0.1){\makebox(0,0){$b$}}
\put(7,0.1){\makebox(0,0){$c$}}
\put(8,0.1){\makebox(0,0){$d$}}
\put(5.65,1.5){\makebox(0,0){$e$}}
\put(6.65,1.5){\makebox(0,0){$f$}}
\put(9,0.1){\makebox(0,0){$x$}}
\put(8.65,1.5){\makebox(0,0){$y$}}
\put(4.1,1){\makebox(0,0){$G$}}
\end{picture}\caption{\textrm{core}$(G)=\{a,e\}$ is a critical set.}%
\label{fig51111}%
\end{figure}

For instance, consider the graph $G$ in Figure \ref{fig51111} and
\[
S_{1}=\left\{  a,e,f,x\right\}  ,S_{2}=\left\{  a,e,c,x\right\}
,S_{3}=\left\{  a,e,d,y\right\}  ,S_{4}=\left\{  a,e,f,y\right\}
,S_{5}=\left\{  a,e,c,y\right\}  .
\]
There is a matching from $V\left(  G\right)  -S_{1}\cup S_{2}\cup
S_{3}=\left\{  b\right\}  $ into $S_{1}\cap S_{2}\cap S_{3}=\left\{
a,e\right\}  $, but there is no matching from $V\left(  G\right)  -S_{1}\cup
S_{4}\cup S_{5}=\left\{  b,d\right\}  $ into $S_{1}\cap S_{4}\cap
S_{5}=\left\{  a,e\right\}  $. Notice that $G$ is not a K\"{o}nig-Egerv\'{a}ry graph.

\newpage

\begin{theorem}
\label{characterization of ke graph}$G$ is a K\"{o}nig-Egerv\'{a}ry graph if
and only if the following conditions hold:

\emph{(i)} $\left\vert \mathrm{corona}(G)\right\vert +\left\vert
\mathrm{core}(G)\right\vert =2\alpha\left(  G\right)  $,

\emph{(ii)} there is a matching from $V\left(  G\right)  -\mathrm{corona}(G)$
into $\mathrm{core}(G)$.
\end{theorem}

\begin{proof}
Let $H=G\left[  \mathrm{corona}(G)\right]  $. Theorem \ref{perfect matching}%
\emph{(ii)} implies that $\alpha\left(  G\right)  =\alpha(H)$.

Suppose that $G$ is a K\"{o}nig-Egerv\'{a}ry graph. Condition \emph{(i)} holds
by Corollary \ref{the fact about KE}.

By Condition \emph{(i)} and Theorem \ref{perfect matching}\emph{(iv)}, $H$ is
also a K\"{o}nig-Egerv\'{a}ry graph. Thus
\[
\left\vert V\left(  G\right)  \right\vert -\mu\left(  G\right)  =\alpha\left(
G\right)  =\alpha(H)=\left\vert V\left(  H\right)  \right\vert -\mu(H),
\]
and, consequently,
\[
\left\vert V\left(  G\right)  \right\vert -\left\vert V\left(  H\right)
\right\vert =\mu\left(  G\right)  -\mu(H).
\]
Let $M$ be a maximum matching of $G$. Now, applying Theorem
\ref{perfect matching}\emph{(ii)} with $\Gamma=\Omega(G)$, we obtain%
\[
\left\vert \mathrm{corona}(G)\right\vert -\left\vert \mathrm{core}%
(G)\right\vert =2\mu(H).
\]
Hence,
\begin{gather*}
\left\vert V\left(  M\right)  -\mathrm{core}(G)\right\vert \leq\left\vert
V\left(  G\right)  -\mathrm{core}(G)\right\vert =\left\vert V\left(  G\right)
-\mathrm{corona}(G)\right\vert +\left\vert \mathrm{corona}(G)-\mathrm{core}%
(G)\right\vert \\
=(\left\vert V\left(  G\right)  \right\vert -\left\vert V\left(  H\right)
\right\vert )+2\mu(H)=(\mu\left(  G\right)  -\mu(H))+2\mu(H)=\mu\left(
G\right)  +\mu(H).
\end{gather*}
Therefore, $\left\vert V\left(  M\right)  -\mathrm{core}(G)\right\vert \leq
\mu\left(  G\right)  +\mu(H)$. But
\[
2\mu\left(  G\right)  =\left\vert V\left(  M\right)  \right\vert =\left\vert
V\left(  M\right)  \cap\mathrm{core}(G)\right\vert +\left\vert V\left(
M\right)  -\mathrm{core}(G)\right\vert .
\]
Thus
\begin{gather*}
\left\vert V\left(  M\right)  \cap\mathrm{core}(G)\right\vert \geq2\mu\left(
G\right)  -(\mu\left(  G\right)  +\mu(H))=\\
=\mu\left(  G\right)  -\mu(H)=\left\vert V\left(  G\right)  \right\vert
-\left\vert V\left(  H\right)  \right\vert =|V\left(  G\right)
-\mathrm{corona}(G)|.
\end{gather*}
Clearly, $M(y)\in V\left(  G\right)  -\mathrm{corona}(G)$ for every $y\in
V\left(  M\right)  \cap\mathrm{core}(G)$. In other words, $M$ induces an
injective mapping, say $M_{1}$, from $V\left(  M\right)  \cap\mathrm{core}(G)$
into $V\left(  G\right)  -\mathrm{corona}(G)$. Since $\left\vert V\left(
M\right)  \cap\mathrm{core}(G)\right\vert \geq|V\left(  G\right)
-\mathrm{corona}(G)|$, we conclude that $M_{1}$ is a bijection. Therefore,
$M_{1}^{-1}$ is a matching from $V\left(  G\right)  -\mathrm{corona}(G)$ into
$V\left(  M\right)  \cap\mathrm{core}(G)$. This completes the proof of
Condition \emph{(ii)}.

Now, suppose that Conditions \emph{(i)} and \emph{(ii)} hold.

In general, $\mu(G)\leq\mu(G-v)+1$ for every vertex $v\in V\left(  G\right)
$. Consequently,$\ $if $H$ is a subgraph of $G$, then $\left\vert V\left(
G\right)  \right\vert -\left\vert V\left(  H\right)  \right\vert \geq
\mu\left(  G\right)  -\mu(H)$.

Condition \emph{(ii)} implies that there exists a matching in $G$ comprised of
a matching $M_{1}$ from $V\left(  G\right)  -\mathrm{corona}(G)$ into
$\mathrm{core}(G)$, and a maximum matching $M_{2}$ of $G\left[
\mathrm{corona}(G)-\mathrm{core}(G)\right]  $. Since there are no edges
connecting $\mathrm{corona}(G)-\mathrm{core}(G)$ and $\mathrm{core}(G)$, we
obtain
\begin{gather*}
\mu\left(  G\right)  \geq\left\vert M_{1}\right\vert +\left\vert
M_{2}\right\vert =\left\vert V\left(  G\right)  -\mathrm{corona}(G)\right\vert
+\mu(G\left[  \mathrm{corona}(G)-\mathrm{core}(G)\right]  )=\\
=\left\vert V\left(  G\right)  \right\vert -\left\vert \mathrm{corona}%
(G)\right\vert +\mu(G\left[  \mathrm{corona}(G)\right]  )=\left\vert V\left(
G\right)  \right\vert -\left\vert V\left(  H\right)  \right\vert +\mu(H).
\end{gather*}
Hence, $\left\vert V\left(  G\right)  \right\vert -\left\vert V\left(
H\right)  \right\vert =\mu\left(  G\right)  -\mu(H)$.

Condition \emph{(i)} together with Theorem \ref{perfect matching}%
\emph{(iii),(iv)} ensure that $H$ is a K\"{o}nig-Egerv\'{a}ry graph, and
$\alpha\left(  H\right)  =\alpha\left(  G\right)  $. Therefore,%
\[
\left\vert V\left(  G\right)  \right\vert -\mu\left(  G\right)  =\left\vert
V\left(  H\right)  \right\vert -\mu(H)=\alpha\left(  H\right)  =\alpha\left(
G\right)  .
\]
Thus $\left\vert V\left(  G\right)  \right\vert =\alpha\left(  G\right)
+\mu(G)$, which means that $G$ is a K\"{o}nig-Egerv\'{a}ry graph as well.
\end{proof}

\begin{remark}
The graphs $G_{1}$ and $G_{2}$ in Figure \ref{fig233} show that none of
Conditions \emph{(i)} or \emph{(ii)} from Theorem
\ref{characterization of ke graph} is enough to infer that $G$ is a
K\"{o}nig-Egerv\'{a}ry graph.
\end{remark}

\begin{corollary}
If $G$ is a K\"{o}nig-Egerv\'{a}ry graph then $\left\vert V\left(  G\right)
-\mathrm{corona}(G)\right\vert \leq\left\vert \mathrm{core}(G)\right\vert $.
\end{corollary}

\section{Conclusions}

In this paper we focus on interconnections between unions and intersections of
maximum independents sets of a graph. Let us say that a family $\emptyset
\neq\Gamma\subseteq\Omega\left(  G\right)  $ is a K\"{o}nig-Egerv\'{a}ry
collection if $\left\vert \bigcup\Gamma\right\vert +\left\vert \bigcap
\Gamma\right\vert =2\alpha(G)$. The set of all K\"{o}nig-Egerv\'{a}ry
collections is denoted as $\Im\left(  G\right)  =\Im\left(  \Omega\left(
G\right)  \right)  $. One of the main findings of this paper can be
interpreted as the claim that $\Im\left(  G\right)  $ is an abstract
simplicial complex for every graph. In other words, every subcollection of a
K\"{o}nig-Egerv\'{a}ry collection is K\"{o}nig-Egerv\'{a}ry as well. We
incline to think that $\Im\left(  G\right)  $ is a new important invariant of
a graph, which may be compared with the nerve of the family of all maximum
independent sets.

Being more specific, we propose the following.

\begin{problem}
Characterize graphs enjoying $\mathrm{core}(G)=\mathrm{nucleus}(G)$.
\end{problem}

\begin{problem}
Characterize graphs satisfying
\[
\left\vert \mathrm{corona}(G)\right\vert +\left\vert \mathrm{core}%
(G)\right\vert =2\left(  \left\vert V\left(  G\right)  \right\vert
-\mu(G)\right)  .
\]

\end{problem}

\begin{conjecture}
If $\left\vert \mathrm{diadem}(G)\right\vert +\left\vert \mathrm{nucleus}%
(G)\right\vert =2\alpha(G)$, then $G$ is a K\"{o}nig-Egerv\'{a}ry graph.
\end{conjecture}


\begin{thebibliography}{99}                                                                                               %
{}

\bibitem {BorosGolLev}E. Boros, M. C. Golumbic, V. E. Levit, \emph{On the
number of vertices belonging to all maximum stable sets of a graph}, Discrete
Applied Mathematics \textbf{124} (2002) 17-25.

\bibitem {ButTruk2007}S. Butenko, S. Trukhanov, \emph{Using critical sets to
solve the maximum independent set problem}, Operations Research Letters
\textbf{35} (2007) 519-524.

\bibitem {DeLaVina}E. DeLaVina, \emph{Written on the Wall II, Conjectures of
Graffiti.pc},\newline http://cms.dt.uh.edu/faculty/delavinae/research/wowII/

\bibitem {Deming1979}R. W. Deming, \emph{Independence numbers of graphs - an
extension of the K\"{o}nig-Egerv\'{a}ry theorem}, Discrete Mathematics
\textbf{27} (1979) 23-33.

\bibitem {GaryJohnson79}M. Garey, D. Johnson, \emph{Computers and
intractability}, W. H. Freeman and Company, New York, 1979.

\bibitem {Korach2006}E. Korach, T. Nguyen, B. Peis, \emph{Subgraph
characterization of red/blue-split graphs and K\"{o}nig-Egerv\'{a}ry graphs},
in: Proceedings of the Seventeenth Annual ACM--SIAM Symposium on Discrete
Algorithms, ACM Press, 2006, 842--850.

\bibitem {Larson2007}C. E. Larson, \emph{A note on critical independence
reductions}, Bulletin of the Institute of Combinatorics and its Applications
\textbf{5} (2007) 34-46.

\bibitem {Larson2011}C. E. Larson, \emph{The critical independence number and
an independence decomposition}, European Journal of Combinatorics \textbf{32}
(2011) 294-300.

\bibitem {LevMan2002a}V. E. Levit, E. Mandrescu, \emph{Combinatorial
properties of the family of maximum stable sets of a graph}, Discrete Applied
Mathematics \textbf{117} (2002) 149-161.

\bibitem {levm4}V. E. Levit, E. Mandrescu, \emph{On }$\alpha^{+}$\emph{-stable
K\"{o}nig-Egerv\'{a}ry graphs}, Discrete Mathematics \textbf{263} (2003) 179-190.

\bibitem {LevMan2012a}V. E. Levit, E. Mandrescu, \emph{Vertices belonging to
all critical independent sets of a graph}, SIAM Journal on Discrete
Mathematics \textbf{26} (2012) 399-403.

\bibitem {LevMan2012b}V. E. Levit, E. Mandrescu, \emph{Critical independent
sets and K\"{o}nig-Egerv\'{a}ry graphs}, Graphs and Combinatorics \textbf{28}
(2012) 243-250.

\bibitem {LevMan2012c}V. E. Levit, E. Mandrescu, \emph{On the core of a
unicyclic graph}, Ars Mathematica Contemporanea \textbf{5} (2012) 325--331.

\bibitem {LevMan2011b}V. E. Levit, E. Mandrescu, \emph{Critical sets in
bipartite graphs}, Annals of Combinatorics \textbf{17} (2013) 543-548.

\bibitem {LevMan2013a}V. E. Levit, E. Mandrescu, \emph{On the structure of the
minimum critical independent set of a graph}, Discrete Mathematics
\textbf{313} (2013) 605-610.

\bibitem {LevMan2013b}V. E. Levit, E. Mandrescu, \emph{On maximum matchings in
K\"{o}nig-Egerv\'{a}ry graphs}, Discrete Applied Mathematics \textbf{161}
(2013) 1635-1638.

\bibitem {LevManLemma2011}V. E. Levit, E. Mandrescu, \emph{A set and
collection lemma}, The Electronic Journal of Combinatorics \textbf{21} (2014) \#P1.40.

\bibitem {LevMan2014}V. E. Levit, E. Mandrescu, On the intersection of all
critical sets of a unicyclic graph, Discrete Applied Mathematics \textbf{162}
(2014) 409-414.

\bibitem {MoskaNobili2014}R. Mosca, P. Nobili, \emph{Polynomial time
recognition of essential graphs having stability number equal to matching
number}, Graphs and Combinatorics (2014). Published online: DOI 10.1007/s00373-014-1483-4

\bibitem {Sterboul1979}F. Sterboul, \emph{A characterization of the graphs in
which the transversal number equals the matching number}, Journal of
Combinatorial Theory B \textbf{27} (1979) 228-229.

\bibitem {Zhang1990}C. Q. Zhang, \emph{Finding critical independent sets and
critical vertex subsets are polynomial problems}, SIAM Journal on Discrete
Mathematics \textbf{3} (1990) 431-438.
\end{thebibliography}
\end{document}